\documentclass[submission,copyright,creativecommons]{eptcs}
\providecommand{\event}{NCMA 2022} 

\usepackage{iftex}
\usepackage{combelow}

\ifpdf
  \usepackage{underscore}         
  \usepackage[T1]{fontenc}        
\else
  \usepackage{breakurl}           
\fi

\usepackage{extarrows,paralist,amssymb,amsmath}
\usepackage{tikz,pgf}\usetikzlibrary{positioning,arrows,automata,backgrounds}

\title{On the Generative Capacity of Contextual Grammars with Strictly Locally Testable Selection Languages}
\author{J\"urgen Dassow
\institute{Fakult\"at f\"ur Informatik, Otto-von-Guericke-Universit\"at Magdeburg, Postfach 4120, 39106 Magdeburg, Germany}
\email{dassow@iws.cs.uni-magdeburg.de} \and
Bianca Truthe
\institute{Institut f\"ur Informatik, Universit\"at Giessen, Arndtstr.~2, 35392 Giessen, Germany}
\email{bianca.truthe@informatik.uni-giessen.de}}
\def\titlerunning{On Contextual Grammars with Strictly Locally Testable Selection Languages}
\def\authorrunning{J.~Dassow, B.~Truthe}

\def\cF{{\cal F}}
\def\cP{{\cal P}}

\def\cEC{{\cal EC}}
\def\cIC{{\cal IC}}

\newcommand{\CIRC}{\mathit{CIRC}}
\newcommand{\COMB}{\mathit{COMB}}
\newcommand{\COMM}{\mathit{COMM}}
\newcommand{\DEF}{\mathit{DEF}}
\newcommand{\FIN}{\mathit{FIN}}

\newcommand{\MON}{\mathit{MON}}
\newcommand{\NC}{\mathit{NC}}
\newcommand{\NIL}{\mathit{NIL}}
\newcommand{\ORD}{\mathit{ORD}}
\newcommand{\PS}{\mathit{PS}}
\newcommand{\REG}{\mathit{REG}}
\newcommand{\SUF}{\mathit{SUF}}
\newcommand{\UF}{\mathit{UF}}

\newcommand{\SLT}{\mathit{SLT}}

\newcommand{\xLra}[2][]{\xLongrightarrow[#1]{#2}}

\newcommand{\lab}[1]{\label{#1}}

\def\Set#1#2{\left\{\: #1\;|\; #2\:\right\}}
\def\set#1#2{\{\: #1\;|\; #2\:\}}
\def\Sets#1{\left\{#1\right\}}
\def\sets#1{\{#1\}}

\def\Setr#1#2{\left\{\: #1\;\left|\; #2\right.\:\right\}}

\def\qmand{\quad \mbox{and} \quad}

\def\slt#1#2#3#4{\textbf{[}#1,#2,#3,#4\textbf{]}}

\tikzstyle{to}=[->, >=stealth]
\tikzstyle{hier}=[->, >=angle 60]
\tikzstyle{hiero}=[->, >=angle 60, dashed]
\tikzstyle{state}=[circle,draw,inner sep=2pt,minimum size=8mm]
\tikzstyle{edgeLabel}=[inner sep=0.5mm,fill=white,text=black]

\newtheorem{theorem}{Theorem}[section]

\newtheorem{lemma}[theorem]{Lemma}

\newtheorem{conjecture}[theorem]{Conjecture}
\newenvironment{proof}{{\em Proof. }}{{}\hspace*{\fill}$\Box$ \par \medskip }
\newenvironment{proof*}{{\em Proof. }}{\par \medskip }

\newlength{\btlabelwidth}
\newlength{\btleftmargin}
\newenvironment{btlists}{\begin{list}{{\rm--}}{%
\setlength{\labelwidth}{\btlabelwidth}\setlength{\leftmargin}{\btleftmargin}%
\setlength{\topsep}{0.7pt plus0.2ex}%
\setlength{\itemsep}{0.7ex plus0.2ex}%
\setlength{\parsep}{0.7pt plus0.2ex}}}{\end{list}}
\newcounter{btlistrmc}

\newcounter{btlistklac}

\begin{document}
\maketitle

\begin{abstract}
We continue the research on the generative capacity of contextual grammars where contexts are adjoined around whole
words (externally) or around subwords (internally) which belong to special regular selection languages. All languages
generated by contextual grammars where all selection languages are elements of a certain subregular language family form
again a language family. We investigate contextual grammars with strictly locally testable selection languages and
compare those families to families which are based on finite, monoidal, nilpotent, combinational, definite,
suffix-closed, ordered, commutative, circular, non-counting, power-separating, or union-free languages.
\end{abstract}

\section{Introduction}

Contextual grammars were introduced by Solomon~Marcus in \cite{Mar69} as a formal model that might be used
for the generation of natural languages.
The derivation steps consist of adjoining contexts to given sentences starting from a finite set. A context
is given by a pair~$(u,v)$ of words. The external adjoining to a word~$x$ gives the word $uxv$ and the
internal adjoining gives all words $x_1ux_2vx_3$ with $x_1x_2x_3=x$. Following the linguistic motivation, conditions are given
for each context which have to be met by the word in order to allow the context to be adjoined. Contextual grammars
where the contexts are adjoined ex- or internally are called external or internal contextual grammars, respectively.
If conditions are given to the subword where a context is to be adjoined, we speak about
external or internal contextual grammars with selection.
Contextual grammars with ex- or internal derivation and selection in a certain family $F$ of languages were defined where
it is required that
the word where a context is wrapped around belongs to a language of the family $F$. Contextual grammars have been studied where
the family $F$ is taken from the Chomsky hierarchy (see \cite{Ist78,Pau98,handbook} and references therein).

The study of external contextual grammars with selection in special regular sets was started by J{\"u}rgen~Dassow in \cite{Das05}.
The research was continued by J{\"u}rgen~Dassow, Florin Manea, and Bianca Truthe
(see \cite{DasManTru11b,DasManTru12a, DasManTru12b,ManTru12})
where further subregular families of selection languages were considered and the effect of subregular selection languages on the
generative power of external and internal contextual grammars was investigated. A recent survey can be found in \cite{Tru21-fi}.
The internal case is different from the case of external contextual grammars, as there are two main differences between the ways
in which words are derived. In the case of internal contextual grammars, it is possible that the insertion of a context into a
sentential form can be done at more than one place, such that the derivation becomes in some sense non-deterministic; in the
case of external grammars, once a context was selected, there is at most one way to insert it: wrapped around the sentential form,
when this word is in the selection language of the context. If a context can be added internally, then it can be added
arbitrarily often (because the subword where the context
is wrapped around does not change) which does not necessarily hold for external grammars.

In the present paper, we investigate the impact of strictly locally testable selection languages in contextual grammars on the
generative capacity and compare it to those of the families which are based on finite, monoidal, nilpotent, combinational, definite,
suffix-closed, ordered, commutative, circular, non-counting, power-separating, or union-free languages.
External contextual grammars with such selection languages have been investigated in \cite{Das-Analele15}. We give here
some further results. Internal contextual grammars with strictly locally testable selection languages have not been
investigated so far. So, this paper gives first results in this area. In the end, we mention some open problems.

\section{Preliminaries}

After giving some notations used in this paper, we first recall the subregular families of languages under investigation
and then recall the contextual grammars with external or internal language generating modes.

We assume that the reader is familiar with the basic concepts of the theory of automata and formal languages. For details,
we refer to~\cite{handbook}.

Given an alphabet $V$, we denote by $V^*$ and $V^+$ the set of all words and the set of all non-empty words over $V$,
respectively. The empty word is denoted by $\lambda$. By $V^k$ and $V^{{}\leq k}$ for some natural number $k$,
we denote the set of all words of the alphabet $V$ with exactly $k$ letters and the set of all words over $V$ with at most $k$
letters, respectively.
For a word $w$, we denote the length of $w$ by $|w|$.

\subsection{Subregular Language Families}

\subsubsection{Definitions of Subregular Language Families}

We consider the following restrictions for regular languages. Let $L$ be a language
over an alphabet $V$. We say that the language $L$ -- with respect to the alphabet $V$ -- is

\begin{itemize}
\item \emph{monoidal} if and only if $L=V^*$,
\item \emph{nilpotent} if and only if it is finite or its complement $V^*\setminus L$ is finite,
\item \emph{combinational} if and only if it has the form
$L=V^*X$
for some subset $X\subseteq V$,
\item \emph{definite} if and only if it can be represented in the form
$L=A\cup V^*B$
where $A$ and $B$ are finite subsets of $V^*$,
\item \emph{suffix-closed} (or \emph{fully initial} or \emph{multiple-entry} language) if
and only if, for any two words~$x\in V^*$ and~$y\in V^*$, the relation $xy\in L$ implies
the relation $y\in L$,
\item \emph{ordered} if and only if the language is accepted by some deterministic finite
automaton
\[{\cal A}=(V,Z,z_0,F,\delta)\]
with an input alphabet $V$, a finite set $Z$ of states, a start state $z_0\in Z$, a set $F\subseteq Z$ of
accepting states and a transition mapping $\delta$ where $(Z,\preceq )$ is a totally ordered set and, for
any input symbol $a\in V$, the relation $z\preceq z'$ implies $\delta (z,a)\preceq \delta (z',a)$,
\item \emph{commutative} if and only if it contains with each word also all permutations of this
word,
\item \emph{circular} if and only if it contains with each word also all circular shifts of this
word,
\item \emph{non-counting} (or \emph{star-free}) if and only if there is a natural
number $k\geq 1$ such that, for any three words $x\in V^*$, $y\in V^*$, and $z\in V^*$, it holds
$xy^kz\in L \mbox{ if and only if } xy^{k+1}z\in L$,
\item \emph{power-separating} if and only if, there is a natural number $m\geq 1$ such that
for any word~$x\in V^*$, either
$J_x^m \cap L = \emptyset$
or
$J_x^m\subseteq L$
where
$J_x^m = \set{ x^n}{n\geq m}$,
\item \emph{union-free} if and only if $L$ can be described by a regular expression which
is only built by product and star,
\item \emph{strictly locally $k$-testable} if and only if there are three subsets $B$, $I$, and $E$ of $V^k$
such that any word $a_1a_2\ldots a_n$ with $n\geq k$ and $a_i\in V$ for $1\leq i\leq n$ belongs to the language $L$
if and only if~$a_1a_2\ldots a_k\in B$, $a_{j+1}a_{j+2}\ldots a_{j+k}\in I$ for $1\leq j\leq n-k-1$,
and $a_{n-k+1}a_{n-k+2}\ldots a_n\in E$,
\item \emph{strictly locally testable} if and only if it is strictly locally $k$-testable for some natural number $k$.
\end{itemize}

We remark that monoidal, nilpotent, combinational, definite, ordered, union-free, and strictly locally~\hbox{($k$-)testable}
languages are regular, whereas non-regular languages of the other types mentioned above exist.
Here, we consider among the commutative, circular, suffix-closed, non-counting,
and power-separating languages only those which are also regular.

Some properties of the languages of the classes mentioned above can be found in
\cite{Shyr91} (monoids),
\cite{GecsegPeak72} (nilpotent languages),
\cite{Ha69} (combinational and commutative languages),
\cite{PerRabSham63} (definite languages),
\cite{GilKou74} and \cite{BrzoJirZou14} (suffix-closed languages),
\cite{ShyThi74-ORD} (ordered languages),
\cite{Das79} (circular languages),
\cite{McNPap71} (non-counting and strictly locally testable languages),
\cite{ShyThi74-PS} (power-separating languages),
\cite{Brzo62} (union-free languages).

By $\FIN$, $\MON$, $\NIL$, $\COMB$, $\DEF$, $\SUF$, $\ORD$, $\COMM$, $\CIRC$, $\NC$, $\PS$, $\UF$, $\SLT_k$ (for
any natural number $k\geq 1$), $\SLT$,
and $\REG$, we denote the families of all finite, monoidal, nilpotent, combinational, definite, regular suffix-closed,
ordered, regular commutative, regular circular, regular non-counting, regular power-separating, union-free,
strictly locally $k$-testable, strictly locally testable, and regular languages, respectively.

A strictly locally testable language characterized by three finite sets $B$, $I$, and $E$ as above which includes
additionally a finite set $F$ of words which are shorter than those of the sets $B$, $I$, and $E$
is denoted by $\slt{B}{I}{E}{F}$.


As the set of all families under consideration, we set
\begin{align*}
\cF &= \sets{ \FIN, \MON, \NIL, \COMB, \DEF, \SUF, \ORD, \COMM, \CIRC, \NC, \PS, \UF, \SLT}
    \cup\set{\SLT_k}{k\geq 1}.
\end{align*}

\subsubsection{Hierarchy of Subregular Language Families}
Many inclusion relations and incomparabilities between these families have been proved in the past, see~\cite{Tru21-fi} 
for a survey. We now insert the families of the strictly locally ($k$-)testable languages into the
existing hierarchy.

The families of strictly locally $k$-testable languages form an
infinite hierarchy of proper inclusions. This is shown in
\cite{SCR-PSP-11} with the witness languages
\[L_h=\{ab^h\}^+\in\SLT_{h+1}\setminus\SLT_h \mbox{ for } h\geq 1.\]
From \cite{McNPap71}, we know the proper inclusion $\SLT\subset\NC$. In
\cite{Das-Analele15}, the proper inclusions $\COMB\subset\SLT_1$
and~$\DEF\subset\SLT$ as well as the incomparability of each
family $\SLT_k$ for $k\geq 1$ with the families $\FIN$, $\NIL$, and~$\DEF$ 
were mentioned but not proved. This will be done in the
sequel. We first give a witness language which will be useful in
all these proofs.\pagebreak

\begin{lemma}\label{l-abna}
Let $L_{\SLT_1,\neg\DEF}=\{a\}\cup\set{ab^na}{n\geq 0}$. Then
$L_{\SLT_1,\neg\DEF}\in\SLT_1\setminus\DEF$.
\end{lemma}
\begin{proof}
The language $L_{\SLT_1,\neg\DEF}$ can be represented as
$\slt{\{a\}}{\{b\}}{\{a\}}{\emptyset}$, hence
$L_{\SLT_1,\neg\DEF}\in\SLT_1$.

Suppose, this language is definite. Then there are two finite
subsets $D_s\subset\{a,b\}^*$ and $D_e\subset\{a,b\}^*$ such that
\[L_{\SLT_1,\neg\DEF}=D_s\cup \{a,b\}^*D_e.\]
Let $k=\max\set{|w|}{w\in D_s\cup D_e}+1$.
The word $ab^ka$ belongs to
the language $L_{\SLT_1,\neg\DEF}$ but not to the subset $D_s$ due
to its length. Hence, $ab^ka\in\{a,b\}^*D_e$ and also
\[ab^ka\in\{a,b\}^+D_e\]
due to the length of the word. Then we
have also $b^{k+1}a\in\{a,b\}^+D_e$ and, therefore,
\[b^{k+1}a\in L_{\SLT_1,\neg\DEF}\]
which is a contradiction. Thus, $L_{\SLT_1,\neg\DEF}\notin\DEF$.
\end{proof}

The language $L_{\SLT_1,\neg\DEF}$ is a witness language for the
properness of the three inclusions stated in the following lemmas.

\begin{lemma}\label{l-mon-slt1}
The proper inclusion $\MON\subset\SLT_1$ holds.
\end{lemma}
\begin{proof}
We first prove that $\MON$ is included in $\SLT_1$. Let $L$ be a
monoidal language over an alphabet $V$. Then $L=V^*$. With
$\slt{V}{V}{V}{\lambda}$, we have a representation of the language
$L$ as a strictly locally 1-testable language. Hence,
$\MON\subseteq\SLT_1$.

A witness language for the properness is the language
$L_{\SLT_1,\neg\DEF}$ which, according to Lemma~\ref{l-abna},
belongs to the class $\SLT_1$ but not to $\DEF$ and not to $\MON$
because $\MON\subseteq\DEF$.
\end{proof}

\begin{lemma}\label{l-comb-slt1}
The proper inclusion $\COMB\subset\SLT_1$ holds.
\end{lemma}
\begin{proof}
We first prove that $\COMB$ is included in $\SLT_1$. Let $L$ be a
combinational language over an alphabet $V$. Then $L=V^*X$ for
some subset $X\subseteq V$. With $\slt{V}{V}{X}{\emptyset}$, we
have a representation of the language $L$ as a strictly locally
1-testable language. Hence, $\COMB\subseteq\SLT_1$.

A witness language for the properness is the language
$L_{\SLT_1,\neg\DEF}$ which, according to Lemma~\ref{l-abna},
belongs to the class $\SLT_1$ but not to $\DEF$ and not to $\COMB$
because $\COMB\subseteq\DEF$.
\end{proof}

\begin{lemma}\label{l-def-slt}
The proper inclusion $\DEF\subset\SLT$ holds.
\end{lemma}
\begin{proof}
We first prove $\DEF\subseteq\SLT$. Let $L$ be a definite language
over an alphabet $V$. Then $L=D_s\cup V^*D_e$ for some finite
subsets $D_s\subset V^*$ and $D_e\subset V^*$. Let
$k=\max\set{|w|}{w\in D_s\cup D_e}+1$. Further, let
\begin{itemize}
\item $F=\set{w}{w\in L\cap V^{{}\leq k-1}}$ be the set of all
words of $L$ with a length smaller than $k$, \item $B=V^k$ and
$I=V^k$ the set of all words over the alphabet $V$ with a length
of $k$, \item $E=V^*D_e\cap V^k$ the set of all words of the set
$V^*D_e$ with length $k$,
\end{itemize}
and $L'$ be the strictly locally $k$-testable language represented
by $\slt{B}{I}{E}{F}$. We now prove that $L=L'$ holds.

We first show $L\subseteq L'$. Let $w\in L$. If $|w|<k$, then
$w\in F$ and, hence, $w\in L'$. Otherwise, $w\in V^*D_e$ and there
are words $w_0$ and $w_1$ such that $w=w_1w_0$ and $w_0\in
V^*D_e\cap V^k$ (the word $w_0$ is the suffix of $w$ of length
$k$). Every subword of $w$ of length $k$ belongs to the set $V^k$.
Hence, the prefix of $w$ of length $k$ belongs to the set $B$,
every proper infix of $w$ of length $k$ belongs to the set $I$,
and the suffix $w_0$ belongs to the set $E$. Therefore, we have
$w\in L'$ also in this case.

We now show $L'\subseteq L$. Let $w\in L'$. If $w\in F$, then
$w\in L\cap  V^{{}\leq k-1}$ and, hence, $w\in L$. Otherwise, the
length $m=|w|$ of $w$ is at least $k$ and the word $w$ is composed
of $m$ letters $x_i\in V$ for $1\leq i\leq m$ such that
$w=x_1x_2\ldots x_m$. Then we have for the prefix $x_1x_2\ldots
x_k\in B$, for each infix $x_{j+1}x_{j+2}\ldots x_{j+k}\in I$
for~$1\leq j\leq m-1-k$, and for the suffix
$x_{m-k+1}x_{m-k+2}\ldots x_m\in E$. Therefore,
$x_{m-k+1}x_{m-k+2}\ldots x_m\in V^*D_e$ and $x_1x_2\ldots
x_{m-k}\in V^*$. Hence, $w\in V^*D_e$ and $w\in L$.

Since $L=L'$ and $L'\in \SLT_k$ by construction, we also have that
$L\in \SLT_k$ and, thus, also $L\in \SLT$.

A witness language for the properness of the inclusion
$\DEF\subseteq\SLT$ is the language $L_{\SLT_1,\neg\DEF}$ which,
according to Lemma~\ref{l-abna}, belongs to the class $\SLT_1$ and
therefore also to $\SLT$ but not to $\DEF$.
\end{proof}

The language $L_{\SLT_1,\neg\DEF}$ from Lemma~\ref{l-abna} serves
also partially for proving the incomparability of the families of
the strictly locally $k$-testable languages with the families of
the finite languages, of the nilpotent languages, and of the
definite languages.

\begin{lemma}\label{l-uncomp-def-sltk}
The classes $\SLT_k$ for $k\geq 1$ are incomparable to the classes
$\FIN$, $\NIL$, and $\DEF$.
\end{lemma}
\begin{proof}
Due to the inclusion relations, it suffices to show that there is
a language in the class $\SLT_1$ (which belongs also to each other
family $\SLT_k$ for $k>1$) but which is not definite (and hence
neither nilpotent nor finite) and that there are languages~$L_k$
(for $k\geq 1$) which are finite (and, hence, nilpotent and
definite) but not strictly locally $k$-testable.

A language for the first case is $L_{\SLT_1,\neg\DEF}$ from
Lemma~\ref{l-abna} since
$L_{\SLT_1,\neg\DEF}\in\SLT_1\setminus\DEF$.

Languages for the other incomparabilities are $L_k=\{a\}^{k+1}$
for $k\geq 1$. Every such language~$L_k$ is finite. Let $k$ be a
natural number. Suppose that the language $L_k$ is also strictly
locally $k$-testable. Then it is represented by
$\slt{\{a\}^k}{\emptyset}{\{a\}^k}{\emptyset}$. But then, we also
have $a^k\in L_k$ which is a contradiction.
Hence,~$L_k\in\FIN\setminus\SLT_k$ for $k\geq 1$.
\end{proof}

The incomparabilies of the families of the strictly locally
($k$-)testable languages to the families $\UF$ of the union-free
languages, $\SUF$ of the suffix-closed languages, $\COMM$ of the
commutative languages, and $\CIRC$ of the circular languages
follow, due to the inclusion relations, from the incomparabilities
of the classes $\UF$, $\SUF$, $\COMM$, and $\CIRC$ to the classes
$\COMB$ of the combinational languages and $\NC$ of the
non-counting languages which were proved in \cite{HolTru15-ncma}.

Regarding the class $\ORD$ of the ordered languages, we give the
following relations without proofs\footnote{Proofs have been found by
be the authors after the acceptance of this paper.}.


%

\begin{lemma}\label{l-ord}
The proper inclusion $\SLT_1\subset\ORD$ holds. The classes
$\SLT_k$ for $k\geq 2$ are incomparable to the class $\ORD$.
\end{lemma}

If we combine these results with those mentioned in \cite{Tru21-fi}, we obtain the following statement.

\begin{figure}[htb]
\centerline{%
\scalebox{0.7}{\begin{tikzpicture}[node distance=16mm and 26mm,on
grid=true, background rectangle/.style=
{
draw=black!80, rounded corners=1ex}, show background rectangle]
\node (REG) {$\REG$}; \node (PS) [below=of REG] {$\PS$}; \node
(NC) [below=of PS] {$\NC$}; \node (d1) [below=of NC] {}; \node
(d2) [below=of d1] {}; \node (ORD) [below=of d2] {$\ORD$}; \node
(DEF) [below=of ORD] {$\DEF$}; \node (COMB) [below right=of DEF]
{$\COMB$}; \node (NIL) [below=of DEF] {$\NIL$}; \node (FIN)
[below=of NIL] {$\FIN$}; \node (SLT1) [right=of DEF] {$\SLT_1$};
\node (SLT2) [above=of SLT1] {$\SLT_2$}; \node (SLTn) [above=of
SLT2] {$\vdots$}; \node (SLT) [above=of SLTn] {$\SLT$}; \node
(SUF) [right=of SLT2] {$\SUF$}; \node (UF) [right=of SUF] {$\UF$};
\node (COMM) [right=of UF] {$\COMM$}; \node (CIRC) [above=of COMM]
{$\CIRC$}; \node (d3) [below=of COMB] {}; \node (MON) [below
right=of COMB] {$\MON$}; \draw[hier] (FIN) to
node[edgeLabel]{\footnotesize{\cite{Wi78}}} (NIL); \draw[hier]
(MON) to
node[edgeLabel,pos=.4]{\footnotesize{\cite{Tru18-TRsubreg}}}
(NIL); \draw[hier] (MON) [bend right=8]to
node[edgeLabel,pos=.4]{\footnotesize{\ref{l-mon-slt1}}} (SLT1);
\draw[hier] (MON) [bend right=15]to
node[edgeLabel,pos=.4]{\footnotesize{\cite{Tru18-TRsubreg}}}
(SUF); \draw[hier] (MON) [bend right=29]to
node[edgeLabel,pos=.4]{\footnotesize{\cite{Tru18-TRsubreg}}}
(COMM); \draw[hier] (MON) [bend right=22]to
node[edgeLabel,pos=.4]{\footnotesize{\cite{Tru18-TRsubreg}}} (UF);
\draw[hier] (SLT1) to
node[edgeLabel,pos=.4]{\footnotesize{\cite{SCR-PSP-11}}} (SLT2);
\draw[hier] (SLT2) to
node[edgeLabel,pos=.4]{\footnotesize{\cite{SCR-PSP-11}}} (SLTn);
\draw[hier] (SLTn) to
node[edgeLabel,pos=.4]{\footnotesize{\cite{SCR-PSP-11}}} (SLT);
\draw[hier] (SLT) [bend right=4]to
node[edgeLabel,pos=.4]{\footnotesize{\cite{McNPap71}}} (NC);
\draw[hier] (DEF) [bend right=4]to
node[edgeLabel,pos=.6]{\footnotesize{\ref{l-def-slt}}} (SLT);
\draw[hier] (NIL) to
node[edgeLabel,pos=.4]{\footnotesize{\cite{Wi78}}} (DEF);
\draw[hier] (COMB) to
node[edgeLabel,pos=.4]{\footnotesize{\cite{Ha69}}} (DEF);
\draw[hier] (COMB) to
node[edgeLabel,pos=.4]{\footnotesize{\ref{l-comb-slt1}}} (SLT1);
\draw[hier] (ORD) to
node[edgeLabel,pos=.4]{\footnotesize{\cite{ShyThi74-ORD}}} (NC);
\draw[hier] (DEF) to
node[edgeLabel,pos=.5]{\footnotesize{\cite{HolTru15-ncma}}} (ORD);
\draw[hier] (NC) to
node[edgeLabel,pos=.4]{\footnotesize{\cite{ShyThi74-PS}}} (PS);
\draw[hier] (PS) to
node[edgeLabel]{\footnotesize{\cite{HolTru15-ncma}}} (REG);
\draw[hier] (SUF) [bend right=22]to
node[edgeLabel]{\footnotesize{\cite{HolTru15-ncma}}} (PS);
\draw[hier] (COMM) to
node[edgeLabel,pos=.4]{\footnotesize{\cite{HolTru15-ncma}}}
(CIRC); \draw[hier] (CIRC) [bend right=25]to
node[edgeLabel]{\footnotesize{\cite{HolTru15-ncma}}} (REG);
\draw[hier] (UF) [bend right=25]to
node[edgeLabel]{\footnotesize{\cite{HolTru15-ncma}}} (REG);
\draw[hier] (SLT1) to
node[edgeLabel,pos=.4]{\footnotesize{\ref{l-ord}}} (ORD);
\end{tikzpicture}}}
\caption{Hierarchy of subregular language
families}\label{subreg-hier-fig3}
\end{figure}

\begin{theorem}\lab{th-struct-hier}
The inclusion relations presented in Figure~\ref{subreg-hier-fig3}
hold. An arrow from an entry~$X$ to an entry~$Y$ depicts the
proper inclusion $X\subset Y$; if two families are not connected
by a directed path, then they are incomparable.
\end{theorem}

An edge label in Figure~\ref{subreg-hier-fig3} refers to a paper
or a lemma in the present paper where the respective inclusion is
proved (it is not necessarily the first paper where the inclusion
is already mentioned). The incomparabilities which are not related
to strictly locally testable languages are proved
in~\cite{Tru18-TRsubreg}.\pagebreak

\subsection{Contextual Grammars}

Let $F\in\cF$ be a family of languages. A contextual grammar with selection in $F$ is
a triple
$G=(V,\cP,A)$
where
\begin{btlists}
\item $V$ denotes an alphabet,
\item $\cP$ is a finite set of pairs $(S,C)$ with a language $S$ over some subset $U$ of the
alphabet $V$ which belongs to the family $F$ with respect to the alphabet $U$ and a finite
set $C\subset V^*\times V^*$,
\item $A$ denotes a finite subset of $V^*$.
\end{btlists}

The set $V$ is called the basic alphabet; for a selection pair $(S,C)\in\cP$, the language~$S$
is called a selection language and the set $C$ is called a set of contexts of the grammar~$G$;
the elements of $A$ are called axioms.

We now define the derivation modes for contextual grammars with selection.

Let $G=(V,\cP,A)$ be a contextual grammar with selection.
The external derivation relation
$\xLra[\mathrm{ex}]{}$
is defined as follows:
a word $x$ derives a word $y$ if and only if there is a
pair $(S,C)\in\cP$ such that $x\in S$ and $y=uxv$ for some pair $(u,v)\in C$.
The internal derivation relation
$\xLra[\mathrm{in}]{}$
is defined as follows: a word $x$ derives a word $y$
if and only if there are words $x_1,x_2,x_3\in V^*$ such that~$x=x_1x_2x_3$ and a pair $(S,C)\in\cP$ such
that $x_2\in S$ and $y=x_1ux_2vx_3$ for some pair~$(u,v)\in C$.

By $\xLra[\alpha]{*}$ we denote the reflexive and transitive closure of the derivation
relation~$\xLra[\alpha]{}$ for~${\alpha}\in\Sets{\mathrm{ex},\mathrm{in}}$.
The language generated externally or internally by the grammar $G$ is defined as
\[L_{\alpha}(G)=\Set{ z }{ x\xLra[\alpha]{*} z \mbox{ for some } x\in A }\]
for ${\alpha}\in\Sets{\mathrm{ex},\mathrm{in}}$. If the derivation mode is known from the
context, we omit the index~$\alpha$.
For a family~$\mathfrak{L}$ of languages, we denote by~$\cEC(\mathfrak{L})$ and $\cIC(\mathfrak{L})$
the family of all languages generated externally and internally, respectively, by contextual grammars
with selection in $\mathfrak{L}$ (where all selection languages belong to the family~$\mathfrak{L}$).

From the definition follows that the subset relation is preserved under the use of contextual
grammars: if we allow more, we do not obtain less.

\begin{lemma}\label{l-context-gramm-monoton}
For any two language classes $X$ and $Y$ with $X\subseteq Y$,
we have the inclusions
\[\cEC(X)\subseteq\cEC(Y) \qmand \cIC(X)\subseteq\cIC(Y).\]
\end{lemma}

\section{Results}

\subsection{External Contextual Grammars}

When we speak about contextual grammars in this subsection, we mean contextual grammars with external derivation
(also called external contextual grammars). A language of an external contextual grammar is a language which is
externally generated.

In \cite{Das05}, contextual grammars were investigated where the selection languages are finite, monoidal, combinational,
definite, nilpotent, commutative, or suffix-closed and a hierarchy of the language families generated was presented.
In the papers \cite{DasManTru11b,DasManTru12a,Tru14-ncma}, results on the power of external
contextual grammars with circular, ordered, union-free, or definite selection languages are given. The language families
generated by such systems were inserted into the hierarchy from \cite{Das05}. Furthermore, subregular language
families $\cF_n$ were considered and integrated which are obtained by restricting to $n$ states, non-terminal symbols,
or production rules to accept or to generate regular languages (\cite{Tru21-fi}).
We consider here only subregular families defined by structural properties (not resources).


We now present a witness language to prove a proper inclusion and incomparabilities regarding ordered languages
as selection languages.

\begin{lemma}\lab{l-ec-35}
Let $L_{\ORD,\neg\SLT}=\sets{a}^*\cup\sets{a}^*\sets{b}\sets{a}^*\cup\sets{c}\sets{a}^*\sets{b}\sets{a}^*\sets{c}$.
Then
\[L_{\ORD,\neg\SLT}\in\cEC(\ORD)\setminus\cEC(\SLT).\]
\end{lemma}
\begin{proof}
The language $L_{\ORD,\neg\SLT}$ can be generated by the contextual grammar
\[\left(\Sets{a,b,c},\Sets{(\sets{a,b}^*,\Sets{(\lambda,a),(a,\lambda)}),(\sets{a}^*\sets{b}\sets{a,b}^*,\Sets{(c,c)})},
          \Sets{\lambda,b}\right)\]
where the selection languages are ordered: For $\sets{a,b}^*$, only one state is needed; the other selection language
is accepted by a deterministic finite automaton where the transition function is given by~$\delta(z_0,a)=z_0$
and~$\delta(z_0,b)=z_1=\delta(z_1,a)=\delta(z_1,b)$.

Assume that the language $L_{\ORD,\neg\SLT}$ can be generated by a contextual grammar with strictly locally testable
selection languages.
The subset $\sets{c}\sets{a}^*\sets{b}\sets{a}^*\sets{c}$ of $L_{\ORD,\neg\SLT}$ is infinite. Therefore, there is an
infinite selection language $S\subseteq\sets{a}^*\sets{b}\sets{a}^*$ which is used to obtain words of the
set~$\sets{c}\sets{a}^*\sets{b}\sets{a}^*\sets{c}$. Hence, to $S$ belongs some context $(u,v)$ with $u\in\sets{c}\sets{a}^*$
and $v\in\sets{a}^*\sets{c}$. If~$S$ is a strictly locally $k$-testable language, then
$S=\slt{B}{I}{E}{F}$ and $a^k\in B\cap I\cap E$. Then, we have also~$a^k\in S$. Therefore, a word from the
set $\sets{c}\sets{a}^*\sets{c}$ is generated which does not belong to the language~$L_{\ORD,\neg\SLT}$. This contradiction
implies that~$L_{\ORD,\neg\SLT}\notin\cEC(\SLT)$.
\end{proof}

We now prove the mentioned proper inclusion.

\begin{lemma}\label{l-ec-slt1-ord}
The proper inclusion $\cEC(\SLT_1)\subset\cEC(\ORD)$ holds.
\end{lemma}
\begin{proof}
The inclusion $\cEC(\SLT_1)\subseteq\cEC(\ORD)$ follows from Lemma~\ref{l-ord} and Lemma~\ref{l-context-gramm-monoton}.
A witness language for the properness is $L_{\ORD,\neg\SLT}\in\cEC(\ORD)\setminus\cEC(\SLT)$ from Lemma~\ref{l-ec-35}.
\end{proof}

Many incomparability results have been published in \cite{Das05,Das-Analele15,Tru21-fi}.
The only open questions are whether the class $\cEC(\ORD)$ is incomparable to the classes $\cEC(\SLT)$
and $\cEC(\SLT_k)$ for $k\geq 2$. We have the following conjecture. If this proves to be true, then
we have the incomparabilities together with Lemma~\ref{l-ec-35}.

\begin{conjecture}
There is a language $L_{\SLT_2,\neg\ORD}\in\cEC(\SLT_2)\setminus\cEC(\ORD)$.
\end{conjecture}

Summarizing, we have the following result.

\begin{figure}[htb]
\centerline{%
\scalebox{0.75}{\begin{tikzpicture}[node distance=18mm and 34mm,on grid=true,
background rectangle/.style=
{
draw=black!80,
rounded corners=1ex},
show background rectangle]
\node (REG) {$\cEC(\REG)\stackrel{\cite{DasManTru12a}}{=}\cEC(\UF)$};
\node (PS) [below=of REG] {$\cEC(\PS)$};
\node (NC) [below=of PS] {$\cEC(\NC)$};
\node (d1) [below=of NC] {};
\node (d2) [below=of d1] {};
\node (ORD) [below=of d2] {$\cEC(\ORD)$};
\node (DEF) [below=of ORD] {$\cEC(\DEF)$};
\node (NIL) [below right=of DEF] {$\cEC(\NIL)$};
\node (COMB) [below=of DEF] {$\cEC(\COMB)$};
\node (MON) [below=of COMB] {$\cEC(\MON)$};
\node (FIN) [below=of MON] {$\cEC(\FIN)$};
\node (COMM) [right=of ORD] {$\cEC(\COMM)$};
\node (CIRC) [above=of COMM] {$\cEC(\CIRC)$};
\node (SLT1) [left=of DEF] {$\cEC(\SLT_1)$};
\node (SLT2) [above=of SLT1] {$\cEC(\SLT_2)$};
\node (SLTn) [above=of SLT2] {$\vdots$};
\node (SLT) [above=of SLTn] {$\cEC(\SLT)$};
\node (SUF) [left=of SLT2] {$\cEC(\SUF)$};
\draw[hier] (FIN) to node[edgeLabel]{\cite{Das05}} (MON);
\draw[hier] (MON) to node[edgeLabel]{\cite{Das-Analele15}} (COMB);
\draw[hier] (MON) [bend right=15]to node[edgeLabel]{\cite{Das05}} (NIL);
\draw[hier] (MON) [bend right=-29]to node[edgeLabel]{\cite{Das05}} (SUF);
\draw[hier] (NIL) to node[edgeLabel,near start]{\cite{Das05}} (DEF);
\draw[hier] (NIL) to node[edgeLabel]{\cite{Das05}} (COMM);
\draw[hier] (COMB) to node[edgeLabel]{\cite{Tru21-fi}} (DEF);
\draw[hiero] (ORD) to (NC);
\draw[hier] (DEF) to node[edgeLabel]{\cite{Tru14-ncma}} (ORD);
\draw[hier] (NC) to node[edgeLabel]{\cite{Tru21-fi}} (PS);
\draw[hier] (PS) to node[edgeLabel]{\cite{Tru21-fi}} (REG);
\draw[hier] (COMM) to node[edgeLabel]{\cite{DasManTru12a}} (CIRC);
\draw[hier] (CIRC) [bend right=25]to node[edgeLabel]{\cite{DasManTru12a}} (REG);
\draw[hier] (SUF) [bend right=-36]to node[edgeLabel]{\cite{Tru21-fi}} (PS);
\draw[hier] (COMB) to node[edgeLabel]{\cite{Das-Analele15}} (SLT1);
\draw[hier] (SLT1) to node[edgeLabel]{\cite{Das-Analele15}} (SLT2);
\draw[hier] (SLT2) to node[edgeLabel]{\cite{Das-Analele15}} (SLTn);
\draw[hier] (SLTn) to node[edgeLabel]{\cite{Das-Analele15}} (SLT);
\draw[hier] (SLT) to node[edgeLabel]{\cite{Das-Analele15}} (NC);
\draw[hier] (DEF) [bend right=-5]to node[edgeLabel]{\cite{Das-Analele15}} (SLT);
\draw[hier] (SLT1) to node[edgeLabel]{\ref{l-ec-slt1-ord}} (ORD);
\end{tikzpicture}}}
\caption{Hierarchy of language families by external contextual grammars with
selection languages defined by structural properties. An edge label refers to
the paper or lemma where the respective inclusion is proved.}\label{ec-fig-struct}
\end{figure}
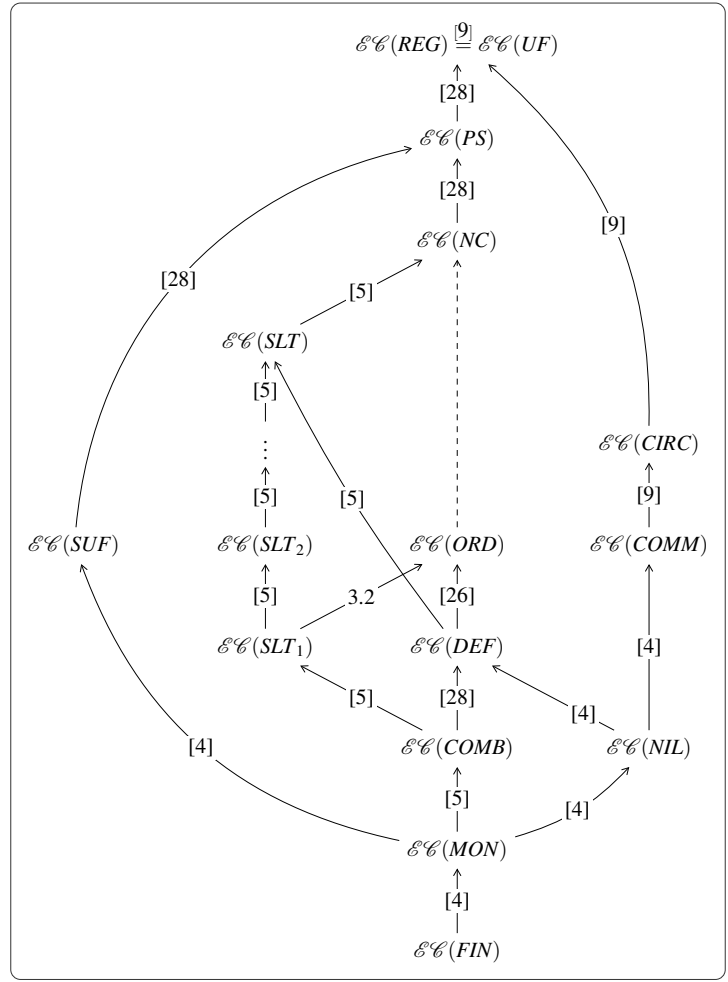

\begin{theorem}\lab{th-ec-struct-hier}
The inclusion relations presented in Figure~\ref{ec-fig-struct} hold.
An arrow from an entry~$X$ to an entry~$Y$ depicts the proper inclusion $X\subset Y$; the dashed arrow
from~$\cEC(\ORD)$ to~$\cEC(\NC)$ indicates that it is not known so far whether the inclusion 
is proper or whether equality holds.
With exception of the pairs~$(\cEC(\ORD),\cEC(\SLT))$ and~$(\cEC(\ORD),\cEC(\SLT_k))$ for $k\geq 2$, if two
families~$X$ and~$Y$ are not connected by a directed path, then~$X$ and~$Y$ are incomparable; in the exceptional
cases, $\cEC(\ORD)\not\subseteq \cEC(\SLT)$ and~$\cEC(\ORD)\not\subseteq \cEC(\SLT_k)$ for $k\geq 2$ hold.
\end{theorem}

\subsection{Internal Contextual Grammars}

When we speak about contextual grammars in this subsection, we mean contextual grammars with
internal derivation (also called internal contextual grammars). A language of an internal
contextual grammar is a language which is internally generated.

In \cite{ManTru12}, such contextual grammars were investigated where the selection
languages belong to families~$\cF_n$ which are obtained by restriction to $n$ states or
$n$ non-terminal symbols, productions, or symbols to accept or to generate regular languages.
In \cite{DasManTru12b}, the research was continued and extended to the investigation of
the effect of finite, monoidal, nilpotent, combinational, definite, ordered, regular
commutative, regular circular, regular suffix-closed, and union-free selection languages
on the generative capacity of internal contextual grammars.
We consider here only subregular families defined by structural properties (not resources).

In contrast to the external derivation mode, contextual grammars can internally apply
a context infinitely often if it can be applied once. If a word contains a subword which
belongs to a selection language, also the word after inserting the context contains a
subword (namely the same as before) which belongs to this selection language. This difference
has as a consequence that finite selection languages not only yield finitely many words
as in the case of contextual grammars working in the external mode.
Another consequence is that `outer' parts of a word do not have to be added at the end
of the derivation process but can be produced at some time whereas `inner' parts can be
`blown up' later. For this reason, the results obtained for external contextual grammars
are not of much help here.

According to Lemma~\ref{l-context-gramm-monoton}, we have the inclusion
$\cIC(X)\subseteq \cIC(Y)$ whenever we have the proper inclusion $X\subset Y$
for two families of languages $X$ and $Y$.

We now present witness languages for proving the properness of the inclusions
\[\cIC(\COMB)\subset\cIC(\SLT_1)\subset \cIC(\SLT_2)\subset\cdots\subset \cIC(\SLT_k)\subset\cdots\subset\cIC(\SLT)\subset\cIC(\NC)\]
and $\cIC(\DEF)\subset\cIC(\SLT)$.

\begin{lemma}\lab{l-ic-32}
Let $L=\set{ac^nbd^n}{n\geq 0}$. Then $L\in \cIC(\SLT_1)\setminus \cIC(\COMB)$.
\end{lemma}
\begin{proof}
The internal contextual grammar
$(\sets{a,b,c,d}, \sets{(\sets{b}^+, \sets{(c,d)})}, \sets{ab})$
with the strictly locally 1-testable selection language $\sets{b}^+$
(which has a representation as $\slt{\sets{b}}{\sets{b}}{\sets{b}}{\emptyset}$) generates the language $L$.
Thus, $L\in \cIC(\SLT_1)$.

Assume that $L=L(G)$ for some internal contextual grammar $G$ with
combinational selection languages. Then, for sufficiently large $n$
(which is larger than the sum of the longest length of axioms in
$G$ and the maximum of $\vert uv\vert$ for contexts $(u,v)$ of
$G$), we have a derivation $x\Longrightarrow ac^nbd^n$. Because~$x\in L$ holds,
the used context $(\alpha,\beta)$ contains no
letter $a$ and no letter $b$ (otherwise, we can produce a word with
more than two occurrences of $a$ or $b$), we have $x=ac^mbd^m$,
$\alpha=c^{n-m}$, $\beta=d^{n-m}$, and the context is wrapped
around a subword $c^tbd^s$ for some numbers $t$ and $s$
with $m\geq t\geq 0$, $m\geq s\geq 0$.
Since the selection language $C$ is combinational, we get
$ac^mbd^s\in C$ by $c^tbd^s\in C$. Therefore, we have the
derivation $x=ac^mbd^sd^{m-s}\Longrightarrow \alpha ac^mbd^s\beta d^{m-s}=c^{n-m}ac^mbd^n$,
i.\,e., we can derive a word not in $L$.
Thus,~$L\notin \cIC(\COMB)$.
\end{proof}

\begin{lemma}\lab{l-ic-33}
Let $n$ be a natural number with $n\geq 2$ and
\[L_n=\set{a^mb^{2n}c^m}{m\geq n}\cup\sets{a^{n-1}b^nc^{n-1}}.\]
Then $L_n\in(\cIC(\SLT_n)\cap\cIC(\FIN))\setminus\cIC(\SLT_{n-1})$.
\end{lemma}
\begin{proof}
Let $n$ be a natural number with $n\geq 2$ and $L_n$ the language mentioned in the claim.
The language~$L_n$ is generated by the contextual grammar
\[(\Sets{a,b,c},\Sets{(\textbf{[}\Sets{a^n},\Sets{a,b,c}^n,\Sets{c^n}\textbf{]},\Sets{(a,c)})},\Sets{a^nb^{2n}c^n,a^{n-1}b^nc^{n-1}})\]
with a selection language from the family $\SLT_n$ and by
\[(\Sets{a,b,c},\Sets{(\sets{b^{2n}},\Sets{(a,c)})},\Sets{a^nb^{2n}c^n,a^{n-1}b^nc^{n-1}})\]
with a finite selection language.

The language $L_n$ is not generated by a contextual grammar where all selection languages belong to the family $\SLT_{n-1}$.
Assume the contrary. Since the subset $\set{a^mb^{2n}c^m}{m\geq n}$ of $L_n$ is infinite, there is a selection
language $S=\slt{B}{I}{E}{F}$ used with a word $a^pb^{2n}c^q$ for
two natural numbers $p\geq 0$ and $q\geq 0$. As $S\in\SLT_{n-1}$, we
have $b^{n-1}\in I$. Then also the word $a^pb^nc^q$ belongs to the selection language
which is a subword of the word $a^{n-1}b^nc^{n-1}\in L_n$.
Hence, another word with exactly $n$ letters $b$ would be generated which is a contradiction to the form
of the words in the language~$L_n$.
\end{proof}

\begin{lemma}\lab{l-ic-34}
Let $L=\Set{a^nb^mc^nd^m}{m\geq 1,\ n\geq 1}$.
Then $L\in \cIC(\SLT_1)\setminus \cIC(\DEF)$.
\end{lemma}
\begin{proof}
The language $L$ can be generated by the contextual grammar
\[G=\left(\Sets{a,b,c,d},\Sets{(S_{ac},\Sets{(a,c)}),(S_{bd},\Sets{(b,d)})},\Sets{abcd}\right)\]
with the strictly locally 1-testable selection languages $S_{ac}=\sets{a}\sets{b}^*\sets{c}$
and $S_{bd}=\sets{b}\sets{c}^*\sets{d}$:
\[S_{ac} = \slt{\sets{a}}{\sets{b}}{\sets{c}}{\emptyset} \quad\mbox{ and }\quad
S_{bd} = \slt{\sets{b}}{\sets{c}}{\sets{d}}{\emptyset}.
\]
Assume that the language $L$ can be generated by a contextual grammar $G'$ with definite selection languages.
Let $S_i=A_i\cup V^*B_i$ for $1\leq i\leq q$ be the selection languages of $G'$. Further, let
\[p=\max \Setr{|w|}{w\in \bigcup_{i=1}^q (A_i\cup B_i)}.\]
Since the language $L$ is infinite and the number of the letters $a$ and $b$ are unbounded in its words, there is a
word $a^rb^sc^rd^s\in L$ with $r\geq p$ and $s\geq p$ such that from this word another one is generated. Hence,
there is a selection language $S_i$ with $1\leq i\leq q$ which contains a word which is a subword of $a^rb^sc^rd^s$.
This word is $a^{r'}b^sc^{r''}$ with $1\leq r'\leq r$ and $1\leq r''\leq r$
or $b^{s'}c^rd^{s''}$ with $1\leq s'\leq s$ and $1\leq s''\leq s$
in order to maintain the form of the words of the language. Since $S_i$ is definite and $s-1+r''\geq p$ or $r-1+s''\geq p$,
the word~$b^{s-1}c^{r''}$ or~$c^{r-1}d^{s''}$ also belongs to the selection language $S_i$. But then a letter $a$ would be
inserted inside the $b$-block or a letter $b$ would be inserted inside the $c$-block. In both cases, a word would be
generated which does not belong to the language $L$. Therefore, the language $L$ cannot be generated by a contextual grammar
with definite selection languages.
\end{proof}

\begin{lemma}\lab{l-ic-35}
Let $L=\set{a^{p_1}ba^{p_2}ba^{p_3+p_1}ba^{p_2}ba^{p_3}}{p_i\geq 1,\ 1\leq i\leq 3}$.
Then $L\in\cIC(\ORD)\setminus\cIC(\SLT)$.
\end{lemma}
\begin{proof}
The language $L$ can be generated by the contextual grammar
\[\left(\Sets{a,b},\Sets{(\sets{a}^*\sets{b}\sets{a}^*\sets{b}\sets{a}^*,\Sets{(a,a)})},
          \sets{ababaababa}\right)\]
where the selection language is ordered since it is accepted by a deterministic finite automaton where the transition
function is given by the following table (the order is $z_0\preceq z_1\preceq z_2\preceq z_3$,
start state is $z_0$, accepting state is~$z_2$):
\[\begin{array}{c|cccc}
\delta & z_0 & z_1 & z_2 & z_3\\\hline
a      & z_0 & z_1 & z_2 & z_3\\
b      & z_1 & z_2 & z_3 & z_3\\
\end{array}\]

Assume that the language $L$ can be generated by a contextual grammar with strictly locally testable selection languages.
The length of each $a$-block is unbounded. Therefore, there is an infinite selection
language $S\subseteq\sets{a}^*\sets{b}\sets{a}^*\sets{b}\sets{a}^*$ used where the length of the $a$-block between the two
letters $b$ is unbounded and which has a context $(a^l,a^l)$ associated to it (otherwise,
a word would be generated which has not the required form of the words of the language $L$).
If $S$ is a strictly locally $k$-testable language, then it contains with a word~$a^qba^rba^s$
with $q\geq 0$, $r\geq k$, and $s\geq 0$ also the word~$a^qba^rba^rba^s$. Adding the context $(a^l,a^l)$ around
such a subword of a word of $L$ would yield a word which does not belong to the language $L$ (a word with a wrong format).
This contradiction implies that~$L\notin\cIC(\SLT)$.
\end{proof}

We now prove the proper inclusions mentioned above.

\begin{theorem}\lab{l-ic-th1}
The relations
\[\cIC(\COMB)\subset\cIC(\SLT_1)\subset \cIC(\SLT_2)\subset\cdots\subset \cIC(\SLT_k)\subset\cdots\subset\cIC(\SLT)\subset\cIC(\NC)\]
and $\cIC(\DEF)\subset\cIC(\SLT)$ hold.
\end{theorem}
\begin{proof}
The inclusions follow from the inclusions of the underlying language families (see \cite{McNPap71,SCR-PSP-11}) and
Lemma~\ref{l-context-gramm-monoton}. The properness is shown by the witness languages in the previous lemmas:
\begin{align*}
\cIC(\COMB)&\subset\cIC(\SLT_1) \mbox{ according to Lemma~\ref{l-ic-32},}\\
\cIC(\SLT_1)&\subset \cIC(\SLT_2)\subset\cdots\subset \cIC(\SLT_k)\subset\cdots\subset\cIC(\SLT)\mbox{ due to Lemma~\ref{l-ic-33},}\\
\cIC(\SLT)&\subset\cIC(\NC) \mbox{ according to Lemma~\ref{l-ic-35}, since $\cIC(\ORD)\subseteq\cIC(\NC)$ (\cite{Tru21-fi}),}\\
\cIC(\DEF)&\subset\cIC(\SLT) \mbox{ due to Lemma~\ref{l-ic-34}, since $\cIC(\SLT_1)\subset\cIC(\SLT)$.}
\end{align*}
\end{proof}

\begin{lemma}\label{l-ic-slt1-ord}
The proper inclusion $\cIC(\SLT_1)\subset\cIC(\ORD)$ holds.
\end{lemma}
\begin{proof}
The inclusion $\cIC(\SLT_1)\subseteq\cIC(\ORD)$ follows from Lemma~\ref{l-ord} and Lemma~\ref{l-context-gramm-monoton}.
A witness language for the properness is the language $L\in\cIC(\ORD)\setminus\cIC(\SLT)$ from Lemma~\ref{l-ic-35}.
\end{proof}

The incomparabilities of the families $\cIC(\COMM)$ and $\cIC(\CIRC)$
with the families $\cIC(\SLT_k)$ for $k\geq 1$ and $\cIC(\SLT)$ follow from
the incomparabilities of the sets $\cIC(\COMM)$ and $\cIC(\CIRC)$ with the sets~$\cIC(\COMB)$ and $\cIC(\NC)$
shown in \cite{Tru21-fi}, since
\[\cIC(\COMB)\subseteq\cIC(\SLT_1)\subseteq\cIC(\SLT_2)\subseteq\cdots\subseteq\cIC(\SLT)\subseteq\cIC(\NC).\]

Regarding $\cIC(\SUF)$, we know that there is a language in the set $\cIC(\COMB)\setminus\cIC(\SUF)$ (\cite{DasManTru12b})
which also belongs to each set $\cIC(\SLT_k)$ for $k\geq 1$ and $\cIC(\SLT)$ due to the inclusion relations. However,
it is still open whether there is a language in the set $\cIC(\SUF)\setminus\cIC(\NC)$ (which would not belong to subsets
of $\cIC(\NC)$ either). So, we cannot use the method as for the classes~$\cIC(\COMM)$ and~$\cIC(\CIRC)$.

In the sequel, we show that, for every number $k\geq 1$, there is a language which belongs to the set~$\cIC(\SUF)$
but not to~$\cIC(\SLT_k)$.
We first note that the internal contextual grammar
\[(\sets{c,d}, \sets{(\sets{c,d}^*, \sets{(c,d)})}, \sets{\lambda})\]
generates the Dyck language $D$ over $\{c,d\}$.
For $k\geq 1$, we set
\begin{align*}
K_k' &= \set{c^{m_0}ac^{m_1}a\dots c^{m_k}ac^{m_{k+1}}bd^{m_0+m_1+\dots +m_{k+1}}}{m_i\geq 0,\ 0\leq i\leq k+1}, \\
K_k''&= K_k' \cup \{c^ka^{2k-1}\}K_k'\sets{d^k},
\end{align*}
and define $K_k$ as the language obtained from $K_k''$ by inserting a word of $D$ at any position.

\begin{lemma}
For all $k\geq 1$, we have $K_k\in \cIC(\SUF)\setminus \cIC(\SLT_k)$.
\end{lemma}
\begin{proof}
The internal contextual grammar
\[(\sets{a,b,c,d}, \sets{(\set{a^rb}{0\leq r\leq k+1}\cup \sets{\lambda}, \sets{(c,d)})}, \sets{a^{k+1}b, c^ka^{3k}bd^k}\]
with a suffix closed selection language generates the language $K_k$. Thus, $K_k\in\cIC(\SUF)$.

Assume that $K_k=L(G)$ for some internal contextual grammar $G$ where all selection languages are strictly locally
$k$-testable. Let $n$ be sufficiently large. Then there is a derivation $x\Longrightarrow c^na^{k+1}bd^n\in K_k$.
Since $x\in K_k$ holds, the used context $(\alpha,\beta)$ contains no letter $a$ and no letter $b$.
We have $x=c^ma^{k+1}bd^m$, $\alpha=c^{n-m}$, $\beta=d^{n-m}$, and the context is wrapped around
a subword $c^ta^{k+1}bd^s$ for some numbers $t$ and $s$ with $m\geq t\geq 0$ and $m\geq s\geq 0$.

Let $t\geq k$. Since the strictly locally $k$-testable
selection language $C$ which is used contains the word~$c^ta^{k+1}bd^s$, it also contains the word $c^ka^{3k}bd^s$.
Analogously, if $s\geq k$, the selection language $C$ also contains the word $c^ta^{3k}bd^k$.
Let $k'=\min\sets{k,t}$ and $k''= \min\sets{k,s}$. Then we have $c^{k'}a^{3k}bd^{k''}\in C$.
Hence, we have the derivation
\[c^ka^{3k}bd^k\Longrightarrow c^{k-k'}\alpha c^{k'}a^{3k}bd^{k''}\beta d^{k-k''}=c^{k+n-m}a^{3k}bd^{k+n-m}\]
which produces a word not in $K_k$.
Therefore, $K_k\notin {\cIC}(\SLT_k)$.
\end{proof}

Together with the result $\cIC(\COMB)\setminus\cIC(\SUF)\not=\emptyset$ (recalled from \cite{DasManTru12b}) which
also implies that~$\cIC(\SLT_k)\setminus\cIC(\SUF)\not=\emptyset$ for all $k\geq 1$, we obtain the following
incomparability result.

\begin{lemma}\label{l-ic-incomp+sltk-suf}
The families $\cIC(\SLT_k)$ for $k\geq 1$ are incomparable to the family $\cIC(\SUF)$.
\end{lemma}

It is left open, whether the family $\cIC(\SUF)$ is also incomparable to the family $\cIC(\SLT)$ or
whether it is a proper subset (since we know already that~$\cIC(\SLT)\setminus\cIC(\SUF)\not=\emptyset$).

Now we investigate the relations of the families
$\cIC(\FIN)$, $\cIC(\NIL)$, and $\cIC(\DEF)$
to the
families $\cIC(\SLT_k)$ for $k\geq 1$
as well as the relation of the family $\cIC(\SLT)$ to the family $\cIC(\ORD)$.

\begin{lemma}\lab{l-ic-incomp+sltk-fin+nil+def}
The families $\cIC(\SLT_k)$ for $k\geq 1$ are incomparable to the families $\cIC(\FIN)$, $\cIC(\NIL)$, and $\cIC(\DEF)$.
\end{lemma}
\begin{proof}
Due to the inclusion relations, it suffices to show that there are languages
\[L_0\in\cIC(\SLT_1)\setminus\cIC(\DEF) \quad\mbox{ and }\quad
L_n\in\cIC(\FIN)\setminus\cIC(\SLT_n) \mbox{ for $n\geq 1$}.\]
From Lemma~\ref{l-ic-34}, we know for $L_0=\Set{a^nb^mc^nd^m}{m\geq 1,\ n\geq 1}$ that
$L_0\in\cIC(\SLT_1)\setminus\cIC(\DEF)$.
From Lemma~\ref{l-ic-33}, we
know $L_n=\set{a^mb^{2n}c^m}{m\geq n}\cup\sets{a^{n-1}b^nc^{n-1}}\in\cIC(\FIN)\setminus\cIC(\SLT_{n-1})$ for~$n\geq 2$.
\end{proof}

From \cite{McNPap71} and by Lemma~\ref{l-context-gramm-monoton}, we know the
inclusion $\cIC(\SLT)\subseteq\cIC(\NC)$; from \cite{Tru21-fi}, we have the relation $\cIC(\ORD)\subseteq\cIC(\NC)$.
Here, we have shown with Lemma~\ref{l-ic-35} that there is a language in the family $\cIC(\ORD)$ which does not belong to
the family~$\cIC(\SLT)$. The question whether the family~$\cIC(\SLT)$ is a proper subset of the family~$\cIC(\ORD)$ or
whether these two families are incomparable is left open.

Summarizing, we have the following result.

\begin{figure}[htb]
\centerline{%
\scalebox{0.75}{\begin{tikzpicture}[node distance=18mm and 34mm,on grid=true,
background rectangle/.style=
{
draw=black!80,
rounded corners=1ex},
show background rectangle]
\node (REG) {$\cIC(\REG)\stackrel{\cite{DasManTru12b}}{=}\cIC(\UF)$};
\node (PS) [below=of REG] {$\cIC(\PS)$};
\node (NC) [below=of PS] {$\cIC(\NC)$};
\node (d1) [below=of NC] {};
\node (d2) [below=of d1] {};
\node (ORD) [below=of d2] {$\cIC(\ORD)$};
\node (DEF) [below=of ORD] {$\cIC(\DEF)$};
\node (COMB) [below=of DEF] {$\cIC(\COMB)$};
\node (NIL) [left=of COMB] {$\cIC(\NIL)$};
\node (MON) [below=of COMB] {$\cIC(\MON)$};
\node (FIN) [left=of MON] {$\cIC(\FIN)$};
\node (SLT1) [right=of DEF] {$\cIC(\SLT_1)$};
\node (SLT2) [above=of SLT1] {$\cIC(\SLT_2)$};
\node (SLTn) [above=of SLT2] {$\vdots$};
\node (SLT) [above=of SLTn] {$\cIC(\SLT)$};
\node (SUF) [right=of SLT2] {$\cIC(\SUF)$};
\node (COMM) [right=of SUF] {$\cIC(\COMM)$};
\node (CIRC) [above=of COMM] {$\cIC(\CIRC)$};
\draw[hier] (FIN) to node[edgeLabel]{\cite{DasManTru12b}} (NIL);
\draw[hier] (MON) to node[edgeLabel]{\cite{DasManTru12b}} (COMB);
\draw[hier] (MON) to node[edgeLabel]{\cite{DasManTru12b}} (NIL);
\draw[hier] (MON) [bend right=25]to node[edgeLabel]{\cite{DasManTru12b}} (SUF);
\draw[hier] (NIL) to node[edgeLabel,near start]{\cite{DasManTru12b}} (DEF);
\draw[hier] (MON) [bend right=30]to node[edgeLabel]{\cite{DasManTru12b}} (COMM);
\draw[hier] (COMB) to node[edgeLabel]{\cite{DasManTru12b}} (DEF);
\draw[hiero] (ORD) to (NC);
\draw[hier] (DEF) to node[edgeLabel]{\cite{Tru14-ncma}} (ORD);
\draw[hier] (NC) to node[edgeLabel]{\cite{Tru21-fi}} (PS);
\draw[hier] (PS) to node[edgeLabel]{\cite{Tru21-fi}} (REG);
\draw[hier] (COMM) to node[edgeLabel]{\cite{DasManTru12b}} (CIRC);
\draw[hier] (CIRC) [bend right=30]to node[edgeLabel]{\cite{DasManTru12b}} (REG);
\draw[hier] (SUF) [bend right=25]to node[edgeLabel]{\cite{Tru21-fi}} (PS);
\draw[hier] (COMB) to node[edgeLabel]{\ref{l-ic-th1}} (SLT1);
\draw[hier] (SLT1) to node[edgeLabel]{\ref{l-ic-th1}} (SLT2);
\draw[hier] (SLT2) to node[edgeLabel]{\ref{l-ic-th1}} (SLTn);
\draw[hier] (SLTn) to node[edgeLabel]{\ref{l-ic-th1}} (SLT);
\draw[hier] (SLT) to node[edgeLabel]{\ref{l-ic-th1}} (NC);
\draw[hier] (DEF) [bend right=5]to node[edgeLabel]{\ref{l-ic-th1}} (SLT);
\draw[hier] (SLT1) to node[edgeLabel]{\ref{l-ic-slt1-ord}} (ORD);
\end{tikzpicture}}}
\caption{Hierarchy of language families by internal contextual grammars with
selection languages defined by structural properties. An edge label
refers to the paper where the respective inclusion is proved.}\label{ic-fig-struct}
\end{figure}
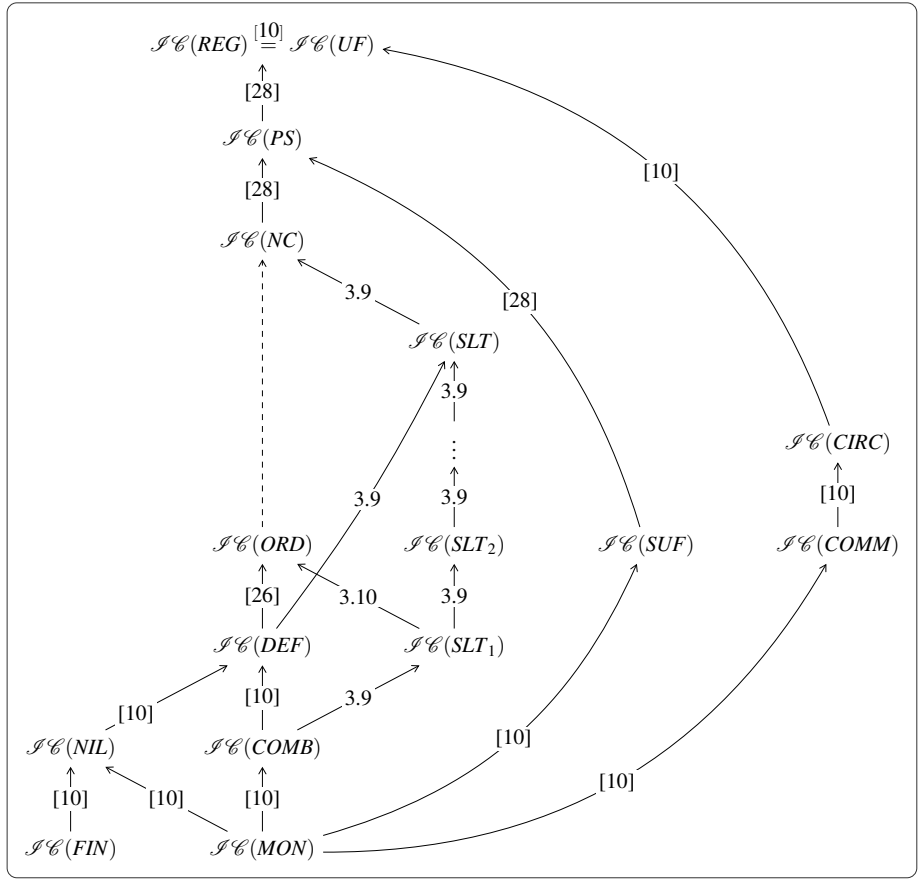

\begin{theorem}\lab{th-ic-struct-hier}
The inclusion relations presented in Figure~\ref{ic-fig-struct} hold.
An arrow from an entry~$X$ to an entry~$Y$ depicts the proper inclusion $X\subset Y$; the dashed arrow
from~$\cIC(\ORD)$ to~$\cIC(\NC)$ indicates that it is not known so far whether the inclusion 
is proper or whether equality holds.
If two families are not connected by a directed path, then they are incomparable with the exception of the
family~$\cIC(\SUF)$ and the families~$\cIC(\ORD)$, $\cIC(\NC)$, and $\cIC(\SLT)$ where
$\cIC(\ORD)\not\subseteq \cIC(\SUF)$, $\cIC(\NC)\not\subseteq \cIC(\SUF)$, and $\cIC(\SLT)\not\subseteq \cIC(\SUF)$ hold,
and with exception of the family $\cIC(\ORD)$ and the families $\cIC(\SLT_k)$ for $k\geq 2$ and~$\cIC(\SLT)$, where
$\cIC(\ORD)\not\subseteq \cIC(\SLT_k)$ for $k\geq 2$ and~$\cIC(\ORD)\not\subseteq \cIC(\SLT)$ hold.
\end{theorem}

\section{Conclusions}

The inclusion relations obtained for the families of languages generated by external or internal contextual grammars
are in most cases the same as for the families where the selection languages are taken from.

For further research, the open questions already mentioned should be considered:
What is the relation between the families~$\cEC(\SLT)$ and~$\cEC(\ORD)$ as well as between
the families~$\cIC(\SLT)$ and~$\cIC(\ORD)$,
especially, is there a language in the set $\cEC(\SLT_2)\setminus\cEC(ORD)$ or a language
in the set~$\cIC(\SLT_2)\setminus\cIC(ORD)$?
Is the family $\cIC(\SUF)$ incomparable to the family $\cIC(\SLT)$ or is it a proper subset?
Additionally, it remains to investigate the relations of the family~$\cIC(\SUF)$ to the
families~$\cIC(\ORD)$ and~$\cIC(\NC)$.

In \cite{Tru21-fi}, two independent
hierarchies have been obtained for each type of contextual grammars, one based on selection languages defined by
structural properties (as considered in this present paper), the other based on resources (number of non-terminal
symbols, production rules, or states). These hierarchies should be merged.

The families of languages which are locally ($k$-)testable (not necessarily in the strict sense) are the Boolean closure
of the families in the strict sense. For contextual grammars where the selection languages are intersections or unions
of strictly locally ($k$-)testable languages, nothing has to be done since the classes $\SLT_k$ for $k\geq 1$ and $\SLT$
are closed under intersection and, for union in a selection pair~$(S_1\cup S_2,C)$, one can take several selection
pairs $(S_1,C)$, $(S_2,C)$ instead. It remains to investigate the impact of locally ($k$-)testable selection languages
which are the complement of a strictly \hbox{locally~($k$-)testable} language.

Additionally, other subfamilies of regular languages could be taken into consideration. Recently, by J\"urgen Dassow,
external contextual grammars have been investigated where the selection languages are ideals or codes~(\cite{Das17,Das21}).
This reseach could be extended to internal contextual grammars with ideals or codes as selection languages.

\nocite{*}
\bibliographystyle{eptcs}
\bibliography{context_gram_subreg}

\begin{thebibliography}{10}
\providecommand{\bibitemdeclare}[2]{}
\providecommand{\surnamestart}{}
\providecommand{\surnameend}{}
\providecommand{\urlprefix}{Available at }
\providecommand{\url}[1]{\texttt{#1}}
\providecommand{\href}[2]{\texttt{#2}}
\providecommand{\urlalt}[2]{\href{#1}{#2}}
\providecommand{\doi}[1]{doi:\urlalt{https://doi.org/#1}{#1}}
\providecommand{\eprint}[1]{arXiv:\urlalt{https://arxiv.org/abs/#1}{#1}}
\providecommand{\bibinfo}[2]{#2}

\bibitemdeclare{phdthesis}{Brzo62}
\bibitem{Brzo62}
\bibinfo{author}{J.~A. \surnamestart Brzozowski\surnameend}
  (\bibinfo{year}{1962}): \emph{\bibinfo{title}{Regular Expression Techniques
  for Sequential Circuits}}.
\newblock Ph.D. thesis, \bibinfo{school}{Princeton University, Princeton, NJ,
  USA}.

\bibitemdeclare{article}{BrzoJirZou14}
\bibitem{BrzoJirZou14}
\bibinfo{author}{J.~A. \surnamestart Brzozowski\surnameend},
  \bibinfo{author}{G.~\surnamestart Jir\'askov\'a\surnameend} \&
  \bibinfo{author}{C.~\surnamestart Zou\surnameend} (\bibinfo{year}{2014}):
  \emph{\bibinfo{title}{Quotient complexity of closed languages}}.
\newblock {\slshape \bibinfo{journal}{Theory of Computing Systems}}
  \bibinfo{volume}{54}, pp. \bibinfo{pages}{277--292},
\doi{10.1007/s00224-013-9515-7}.

\bibitemdeclare{article}{Das79}
\bibitem{Das79}
\bibinfo{author}{J{\"u}rgen \surnamestart Dassow\surnameend}
  (\bibinfo{year}{1979}): \emph{\bibinfo{title}{On the Circular Closure of
  Languages}}.
\newblock {\slshape \bibinfo{journal}{Elektronische Informationsverarbeitung
  und Kybernetik/Journal of Information Processing and Cybernetics}}
  \bibinfo{volume}{15}(\bibinfo{number}{1--2}), pp. \bibinfo{pages}{87--94}.

\bibitemdeclare{article}{Das05}
\bibitem{Das05}
\bibinfo{author}{J{\"u}rgen \surnamestart Dassow\surnameend}
  (\bibinfo{year}{2005}): \emph{\bibinfo{title}{{Contextual grammars with
  subregular choice}}}.
\newblock {\slshape \bibinfo{journal}{Fundamenta Informaticae}}
  \bibinfo{volume}{64}(\bibinfo{number}{1--4}), pp. \bibinfo{pages}{109--118}.

\bibitemdeclare{article}{Das-Analele15}
\bibitem{Das-Analele15}
\bibinfo{author}{J{\"u}rgen \surnamestart Dassow\surnameend}
  (\bibinfo{year}{2015}): \emph{\bibinfo{title}{Contextual languages with
  strictly locally testable and star free selection languages}}.
\newblock {\slshape \bibinfo{journal}{Analele Universit\u{a}\cb{t}ii
  Bucure\cb{s}ti}} \bibinfo{volume}{62}, pp. \bibinfo{pages}{25--36}.

\bibitemdeclare{article}{Das17}
\bibitem{Das17}
\bibinfo{author}{J{\"u}rgen \surnamestart Dassow\surnameend}
  (\bibinfo{year}{2018}): \emph{\bibinfo{title}{Grammars with control by ideals
  and codes}}.
\newblock {\slshape \bibinfo{journal}{Journal of Automata, Languages and
  Combinatorics}} \bibinfo{volume}{23}(\bibinfo{number}{1--3}), pp.
  \bibinfo{pages}{143--164}, \doi{10.25596/jalc-2018-143}.

\bibitemdeclare{article}{Das21}
\bibitem{Das21}
\bibinfo{author}{J{\"u}rgen \surnamestart Dassow\surnameend}
  (\bibinfo{year}{2021}): \emph{\bibinfo{title}{Remarks on external contextual
  grammars with selection}}.
\newblock {\slshape \bibinfo{journal}{Theoretical Computer Science}}
  \bibinfo{volume}{862}, pp. \bibinfo{pages}{119--129},
\doi{10.1016/j.tcs.2020.07.028}.

\bibitemdeclare{inproceedings}{DasManTru11b}
\bibitem{DasManTru11b}
\bibinfo{author}{J{\"u}rgen \surnamestart Dassow\surnameend},
  \bibinfo{author}{Florin \surnamestart Manea\surnameend} \&
  \bibinfo{author}{Bianca \surnamestart Truthe\surnameend}
  (\bibinfo{year}{2011}): \emph{\bibinfo{title}{{On Contextual Grammars with
  Subregular Selection Languages}}}.
\newblock In \bibinfo{editor}{Markus \surnamestart Holzer\surnameend},
  \bibinfo{editor}{Martin \surnamestart Kutrib\surnameend} \&
  \bibinfo{editor}{Giovanni \surnamestart Pighizzini\surnameend}, editors:
  {\slshape \bibinfo{booktitle}{Descriptional Complexity of Formal Systems --
  13th International Workshop, DCFS 2011, Gie{\ss}en/Limburg, Germany, July
  25\,--\,27, 2011. Proceedings}}, {\slshape \bibinfo{series}{LNCS}}
  \bibinfo{volume}{6808}, \bibinfo{publisher}{Springer-Verlag}, pp.
  \bibinfo{pages}{135--146}, \doi{10.1007/978-3-642-22600-7_11}.

\bibitemdeclare{article}{DasManTru12a}
\bibitem{DasManTru12a}
\bibinfo{author}{J{\"u}rgen \surnamestart Dassow\surnameend},
  \bibinfo{author}{Florin \surnamestart Manea\surnameend} \&
  \bibinfo{author}{Bianca \surnamestart Truthe\surnameend}
  (\bibinfo{year}{2012}): \emph{\bibinfo{title}{{On External Contextual
  Grammars with Subregular Selection Languages}}}.
\newblock {\slshape \bibinfo{journal}{Theoretical Computer Science}}
  \bibinfo{volume}{449}(\bibinfo{number}{1}), pp. \bibinfo{pages}{64--73},
\doi{10.1016/j.tcs.2012.04.008}.

\bibitemdeclare{article}{DasManTru12b}
\bibitem{DasManTru12b}
\bibinfo{author}{J{\"u}rgen \surnamestart Dassow\surnameend},
  \bibinfo{author}{Florin \surnamestart Manea\surnameend} \&
  \bibinfo{author}{Bianca \surnamestart Truthe\surnameend}
  (\bibinfo{year}{2012}): \emph{\bibinfo{title}{On Subregular Selection
  Languages in Internal Contextual Grammars}}.
\newblock {\slshape \bibinfo{journal}{Journal of Automata, Languages, and
  Combinatorics}} \bibinfo{volume}{17}(\bibinfo{number}{2--4}), pp.
  \bibinfo{pages}{145--164}, \doi{10.25596/jalc-2012-145}.

\bibitemdeclare{book}{GecsegPeak72}
\bibitem{GecsegPeak72}
\bibinfo{author}{F.~\surnamestart G\'ecseg\surnameend} \&
  \bibinfo{author}{I.~\surnamestart Peak\surnameend} (\bibinfo{year}{1972}):
  \emph{\bibinfo{title}{Algebraic Theory of Automata}}.
\newblock \bibinfo{publisher}{Academiai Kiado, Budapest}.

\bibitemdeclare{article}{GilKou74}
\bibitem{GilKou74}
\bibinfo{author}{A.~\surnamestart Gill\surnameend} \& \bibinfo{author}{L.~T.
  \surnamestart Kou\surnameend} (\bibinfo{year}{1974}):
  \emph{\bibinfo{title}{Multiple-entry finite automata}}.
\newblock {\slshape \bibinfo{journal}{Journal of Computer and System Sciences}}
  \bibinfo{volume}{9}(\bibinfo{number}{1}), pp. \bibinfo{pages}{1--19},
\doi{10.1016/S0022-0000(74)80034-6}.

\bibitemdeclare{article}{Ha69}
\bibitem{Ha69}
\bibinfo{author}{Ivan~M. \surnamestart Havel\surnameend}
  (\bibinfo{year}{1969}): \emph{\bibinfo{title}{The theory of regular events
  {II}}}.
\newblock {\slshape \bibinfo{journal}{Kybernetika}}
  \bibinfo{volume}{5}(\bibinfo{number}{6}), pp. \bibinfo{pages}{520--544}.

\bibitemdeclare{inproceedings}{HolTru15-ncma}
\bibitem{HolTru15-ncma}
\bibinfo{author}{Markus \surnamestart Holzer\surnameend} \&
  \bibinfo{author}{Bianca \surnamestart Truthe\surnameend}
  (\bibinfo{year}{2015}): \emph{\bibinfo{title}{{On Relations Between Some
  Subregular Language Families}}}.
\newblock In \bibinfo{editor}{Rudolf \surnamestart Freund\surnameend},
  \bibinfo{editor}{Markus \surnamestart Holzer\surnameend},
  \bibinfo{editor}{Nelma \surnamestart Moreira\surnameend} \&
  \bibinfo{editor}{Rog{\'{e}}rio \surnamestart Reis\surnameend}, editors:
  {\slshape \bibinfo{booktitle}{{Seventh Workshop on Non-Classical Models of
  Automata and Applications (NCMA), Porto, Portugal, August 31 -- September 1,
  2015, Proceedings}}}, {\slshape \bibinfo{series}{books@ocg.at}}
  \bibinfo{volume}{318}, \bibinfo{publisher}{{\"O}sterreichische Computer
  Gesellschaft}, pp. \bibinfo{pages}{109--124}.

\bibitemdeclare{article}{Ist78}
\bibitem{Ist78}
\bibinfo{author}{Sorin \surnamestart Istrail\surnameend}
  (\bibinfo{year}{1978}): \emph{\bibinfo{title}{Gramatici contextuale cu
  selectiva regulata}}.
\newblock {\slshape \bibinfo{journal}{Stud. Cerc. Mat.}} \bibinfo{volume}{30},
  pp. \bibinfo{pages}{287--294}.

\bibitemdeclare{inproceedings}{ManTru12}
\bibitem{ManTru12}
\bibinfo{author}{Florin \surnamestart Manea\surnameend} \&
  \bibinfo{author}{Bianca \surnamestart Truthe\surnameend}
  (\bibinfo{year}{2012}): \emph{\bibinfo{title}{{On Internal Contextual
  Grammars with Subregular Selection Languages}}}.
\newblock In \bibinfo{editor}{Martin \surnamestart Kutrib\surnameend},
  \bibinfo{editor}{Nelma \surnamestart Moreira\surnameend} \&
  \bibinfo{editor}{Rog\'erio \surnamestart Reis\surnameend}, editors: {\slshape
  \bibinfo{booktitle}{Descriptional Complexity of Formal Systems -- 14th
  International Workshop, DCFS 2012, Braga, Portugal, July 23\,--\,25, 2012.
  Proceedings}}, {\slshape \bibinfo{series}{LNCS}} \bibinfo{volume}{7386},
  \bibinfo{publisher}{Springer-Verlag}, pp. \bibinfo{pages}{222--235}, \doi{10.1007/978-3-642-31623-4_17}.

\bibitemdeclare{article}{Mar69}
\bibitem{Mar69}
\bibinfo{author}{Solomon \surnamestart Marcus\surnameend}
  (\bibinfo{year}{1969}): \emph{\bibinfo{title}{Contextual grammars}}.
\newblock {\slshape \bibinfo{journal}{Revue Roum. Math. Pures Appl.}}
  \bibinfo{volume}{14}, pp. \bibinfo{pages}{1525--1534}.

\bibitemdeclare{book}{McNPap71}
\bibitem{McNPap71}
\bibinfo{author}{Robert \surnamestart McNaughton\surnameend} \&
  \bibinfo{author}{Seymour \surnamestart Papert\surnameend}
  (\bibinfo{year}{1971}): \emph{\bibinfo{title}{Counter-free Automata}}.
\newblock \bibinfo{publisher}{MIT Press, Cambridge, USA}.

\bibitemdeclare{article}{PerRabSham63}
\bibitem{PerRabSham63}
\bibinfo{author}{M.~\surnamestart Perles\surnameend}, \bibinfo{author}{M.~M.
  \surnamestart Rabin\surnameend} \& \bibinfo{author}{E.~\surnamestart
  Shamir\surnameend} (\bibinfo{year}{1963}): \emph{\bibinfo{title}{The theory
  of definite automata}}.
\newblock {\slshape \bibinfo{journal}{IEEE Trans. Electronic Computers}}
  \bibinfo{volume}{12}, pp. \bibinfo{pages}{233--243},
\doi{10.1109/PGEC.1963.263534}.

\bibitemdeclare{book}{Pau98}
\bibitem{Pau98}
\bibinfo{author}{Gheorghe \surnamestart P\u{a}un\surnameend}
  (\bibinfo{year}{1998}): \emph{\bibinfo{title}{Marcus Contextual Grammars}}.
\newblock \bibinfo{publisher}{Kluwer Publ. House, Doordrecht}, \doi{10.1007/978-94-015-8969-7}.

\bibitemdeclare{inproceedings}{SCR-PSP-11}
\bibitem{SCR-PSP-11}
\bibinfo{author}{Stefano~Crespi \surnamestart Reghizzi\surnameend} \&
  \bibinfo{author}{Pierluigi~San \surnamestart Pietro\surnameend}
  (\bibinfo{year}{2011}): \emph{\bibinfo{title}{From Regular to Strictly
  Locally Testable Languages}}.
\newblock In \bibinfo{editor}{P.~\surnamestart Ambro{\v{z}}\surnameend},
  \bibinfo{editor}{{\v{S}}.~\surnamestart Holub\surnameend} \&
  \bibinfo{editor}{Z.~\surnamestart Mas{\'{a}}kov{\'{a}}\surnameend}, editors:
  {\slshape \bibinfo{booktitle}{8th International Conference WORDS 2011}},
  {\slshape \bibinfo{series}{EPTCS}}~\bibinfo{volume}{63}, pp.
  \bibinfo{pages}{103--111}, \doi{10.4204/EPTCS.63.14}.

\bibitemdeclare{book}{handbook}
\bibitem{handbook}
\bibinfo{editor}{Grzegorz \surnamestart Rozenberg\surnameend} \&
  \bibinfo{editor}{Arto \surnamestart Salomaa\surnameend}, editors
  (\bibinfo{year}{1997}): \emph{\bibinfo{title}{Handbook of Formal Languages}}.
\newblock \bibinfo{publisher}{Springer-Verlag, Berlin}, \doi{10.1007/978-3-642-59136-5}.

\bibitemdeclare{book}{Shyr91}
\bibitem{Shyr91}
\bibinfo{author}{H.~J. \surnamestart Shyr\surnameend} (\bibinfo{year}{1991}):
  \emph{\bibinfo{title}{Free Monoids and Languages}}.
\newblock \bibinfo{publisher}{Hon Min Book Co., Taichung, Taiwan}.

\bibitemdeclare{article}{ShyThi74-ORD}
\bibitem{ShyThi74-ORD}
\bibinfo{author}{H.~J. \surnamestart Shyr\surnameend} \&
  \bibinfo{author}{G.~\surnamestart Thierrin\surnameend}
  (\bibinfo{year}{1974}): \emph{\bibinfo{title}{{Ordered Automata and
  Associated Languages}}}.
\newblock {\slshape \bibinfo{journal}{Tankang Journal of Mathematics}}
  \bibinfo{volume}{5}(\bibinfo{number}{1}), pp. \bibinfo{pages}{9--20}.

\bibitemdeclare{article}{ShyThi74-PS}
\bibitem{ShyThi74-PS}
\bibinfo{author}{H.~J. \surnamestart Shyr\surnameend} \&
  \bibinfo{author}{G.~\surnamestart Thierrin\surnameend}
  (\bibinfo{year}{1974}): \emph{\bibinfo{title}{{Power-Separating Regular
  Languages}}}.
\newblock {\slshape \bibinfo{journal}{Mathematical Systems Theory}}
  \bibinfo{volume}{8}(\bibinfo{number}{1}), pp. \bibinfo{pages}{90--95},
\doi{10.1007/BF01761710}.

\bibitemdeclare{inproceedings}{Tru14-ncma}
\bibitem{Tru14-ncma}
\bibinfo{author}{Bianca \surnamestart Truthe\surnameend}
  (\bibinfo{year}{2014}): \emph{\bibinfo{title}{{A Relation Between Definite
  and Ordered Finite Automata}}}.
\newblock In \bibinfo{editor}{Suna \surnamestart Bensch\surnameend},
  \bibinfo{editor}{Rudolf \surnamestart Freund\surnameend} \&
  \bibinfo{editor}{Friedrich \surnamestart Otto\surnameend}, editors: {\slshape
  \bibinfo{booktitle}{{Sixth Workshop on Non-Classical Models of Automata and
  Applications (NCMA), Kassel, Germany, July 28--29, 2014, Proceedings}}},
  {\slshape \bibinfo{series}{books@ocg.at}} \bibinfo{volume}{304},
  \bibinfo{publisher}{{\"O}sterreichische Computer Gesellschaft}, pp.
  \bibinfo{pages}{235--247}.

\bibitemdeclare{techreport}{Tru18-TRsubreg}
\bibitem{Tru18-TRsubreg}
\bibinfo{author}{Bianca \surnamestart Truthe\surnameend}
  (\bibinfo{year}{2018}): \emph{\bibinfo{title}{Hierarchy of Subregular
  Language Families}}.
\newblock \bibinfo{type}{Technical Report},
  \bibinfo{institution}{Justus-Liebig-Universit{\"{a}}t Giessen, Institut
  f{\"{u}}r Informatik, IFIG Research Report 1801}.

\bibitemdeclare{article}{Tru21-fi}
\bibitem{Tru21-fi}
\bibinfo{author}{Bianca \surnamestart Truthe\surnameend}
  (\bibinfo{year}{2021}): \emph{\bibinfo{title}{Generative Capacity of
  Contextual Grammars with Subregular Selection Languages}}.
\newblock {\slshape \bibinfo{journal}{Fundamenta Informaticae}}
  \bibinfo{volume}{180}, pp. \bibinfo{pages}{1--28},
\doi{10.3233/FI-2021-2037}.

\bibitemdeclare{mastersthesis}{Wi78}
\bibitem{Wi78}
\bibinfo{author}{Barbara \surnamestart Wiedemann\surnameend}
  (\bibinfo{year}{1978}): \emph{\bibinfo{title}{{Vergleich der
  Leistungsf\"ahigkeit endlicher determinierter Automaten}}}.
\newblock \bibinfo{type}{Diplom\-arbeit}, \bibinfo{school}{Universit\"at
  Rostock}.

\end{thebibliography}

\end{document}